\documentclass[11pt]{article}
\usepackage[T2A]{fontenc}
\usepackage[english]{babel}
\usepackage{ alphalph,etoolbox }
\usepackage{amsthm, amsmath, amssymb, amsbsy, amscd, amsfonts, latexsym, euscript,epsf, bm}
\newtheorem{thm}{Theorem}[section]
\usepackage{bbold, bm}
\usepackage{epic,eepic}
\usepackage{texdraw}
\usepackage[dvips]{color}
\usepackage{color}
\usepackage{ dsfont }
\usepackage{graphicx}
\usepackage{exscale,relsize}
\usepackage{authblk}
\usepackage{url}
\usepackage{enumerate}
\usepackage{graphicx}

\newtheorem{remark}[thm]{Remark}

\newtheorem{counterexample}[thm]{Counterexample}
\newtheorem{theorem}[thm]{Theorem}

\DeclareMathOperator{\tr}{tr}
\DeclareMathOperator{\id}{id}

\title{From NLS type matrix refactorisation problems to set-theoretical solutions  of  the 2-  and  3-simplex equations}
\date{}

\author[1]{S. Konstantinou-Rizos\thanks{skonstantin84@gmail.com}}
\affil[1]{Centre of Integrable Systems, P.G. Demidov Yaroslavl State University, Yaroslavl, Russia}
\patchcmd{\subequations}{\alph{equation}}{\alphalph{\value{equation}}}{}{}

\begin{document}

\maketitle

\begin{abstract}
We present a method for constructing hierarchies of solutions to $n$-simplex equations by variating the spectral parameter in  their Lax  representation. We use this method to derive new solutions to the set-theoretical 2-  and  3-simplex equations which are related to the Adler map and Nonlinear Schr\"odinger (NLS) type equations. Moreover, we prove that some of the derived Yang--Baxter maps are completely integrable.
\end{abstract}

\bigskip

\begin{quotation}
\noindent{\bf PACS numbers:}
02.30.Ik, 02.90.+p, 03.65.Fd.
\end{quotation}

\begin{quotation}
\noindent{\bf MSC 2020:}
35Q55, 16T25.
\end{quotation}

\begin{quotation}
\noindent{\bf Keywords:} Yang--Baxter maps, Zamolodchikov tetrahedron equation, NLS equation, Adler map, NLS type Yang--Baxter maps, variation of the spectral parameter.
\end{quotation}

\allowdisplaybreaks

\section{Introduction}
By $n$-simplex maps we mean  solutions to the set-theoretical $n$-simplex equations which are important equations of Mathematical Physics. These  equations find applications in various fields of mathematics, especially in the theory of integrable systems.

The Yang--Baxter and tetrahedron maps are solutions to the set-theoretical 2-simplex (namely, the Yang--Baxter) and 3-simplex (namely, the Zamolodchikov tetrahedron) equations, respectively.

One of the most important applications of $n$-simplex maps is  their relation to integrable systems of partial differential equations (PDEs) and partial difference equations (P$\Delta$Es). Known connections include the derivation of $n$-simplex maps via the symmetries of integrable P$\Delta$Es \cite{Kassotakis-Tetrahedron, Pap-Tongas, Pap-Tongas-Veselov}, the construction of integrable P$\Delta$Es using separable invariants of Yang--Baxter maps \cite{Pavlos-Maciej, Pavlos-Maciej-2}, and the construction of Yang--Baxter maps via Darboux transformations \cite{Giota-Miky, GKM, Sokor-2020, Sokor-Sasha, Papamikos-JP-Sasha} and their linearisations \cite{BIKRP, IKKRP}.

This paper focuses on methods for constructing Yang--Baxter and tetrahedron maps. The most, probably, important construction of Yang--Baxter and Zamolodchikov maps is from the local  1-simplex \cite{Suris} and 2-simplex \cite{Kashaev-Sergeev, Nijhoff} equations, respectively.

It has become recently understood that correspondences satisfying the local $(n-1)$-simplex equations lead to $n$-simplex maps \cite{Igonin-Sokor}, although, $n$-simplex maps are usually equivalent to the local $(n-1)$-simplex equation. Furthermore, in \cite{Sokor-Sasha} we showed how to construct Yang--Baxter maps that are equivalent to Darboux-matrix refactorisation problems, whereas in  \cite{Sokor-2020-2} we constructed tetrahedron maps by considering the local Yang--Baxter equation for certain Darboux  matrices.

In this paper, we combine and extend the ideas of the works \cite{Igonin-Sokor},  \cite{Sokor-Sasha} and \cite{Sokor-2020-2}. In particular, we replace the spectral parameter of the Darboux matrices that participate in matrix refactorisation problems by a function. Specifically, in the case  when, for a particular Darboux matrix, the local 1-simplex equation is equivalent to a Yang--Baxter map, the variation of the spectral parameter will lead to a correspondence that allows the construction of several new Yang--Baxter maps. On the other hand, as  it was shown in  \cite{Sokor-2020-2}, the Darboux matrices do not solve the local Yang--Baxter equation for arbitrary values of the spectral parameter, and the spectral parameter was set to zero. The variation of the spectral parameter will lead to correspondences that define new tetrahedron maps. 

To summarise, we list the new results  contained in this paper:
\begin{itemize}
    \item we present a counter example showing that a unique solution to the local 1-simplex equation is not necessarily a Yang--Baxter map;
    \item we present an approach for constructing new solutions to the 2- and 3-simplex equations starting from  the Lax representation of known solutions;
    \item  we derive new Yang--Baxter maps of Adler and NLS type and a Zamolodchikov tetrahedron map of NLS type;
    \item we prove Liouville integrability for some of the derived  maps.
\end{itemize}

This article is structured  as follows: in the next section, we explain what the Yang--Baxter and tetrahedron maps are, and also we define the equations that generate them (local 1- and 2-simplex equations). In section \ref{parameter variation}, we describe the method of variation of the spectral parameter and apply the method to the famous Adler map. We construct a new, Liouville integrable, Yang--Baxter map. In section \eqref{Yang--Baxter maps}, we apply the variation of the spectral parameter to the Darboux transformations for the NLS and the derivative NLS equations in order  to  derive new parametric Yang--Baxter maps. We prove that some of the derived Yang--Baxter maps are completely (Liouville) integrable. Section \ref{tetrahedron maps} deals with the construction of a new parametric tetrahedron map of NLS  type using again the variation of the spectral parameter. Finally,  in section \ref{conclusions}, we list some tasks for future work.

\section{Yang--Baxter and tetrahedron maps}\label{prel}
\subsection{Yang--Baxter and local 1-simplex equation}
Let $\mathcal{A}$ be a set, and $\mathcal{A}^n=\underbrace{\mathcal{A}\times \mathcal{A}\times\dots\times \mathcal{A}}_{n}$. A map $Y:\mathcal{A}^2\rightarrow \mathcal{A}^2$, i.e. $Y(x,y)=(u(x,y),v(x,y))$, is a Yang--Baxter map if it solves the equation \cite{Drinfeld}:
\begin{equation}\label{eq_YB}
Y^{12}\circ Y^{13}\circ Y^{23}=Y^{23}\circ Y^{13}\circ Y^{12}.
\end{equation}
This is  the set-theoretical 2-simplex (Yang--Baxter) equation.
Functions $Y_{ij}:\mathcal{A}^3\rightarrow\mathcal{A}^3$  are defined as:
$$
Y_{12}(x,y,z)=\big(u(x,y),v(x,y),z\big),~~
Y_{23}(x,y,z)=\big(x,u(y,z),v(y,z)\big),~~
Y_{13}(x,y,z)=\big(u(x,z),y,v(x,z)\big),
$$

Let $a$ and  $b$ be  complex parameters. A parametric Yang--Baxter map is a map $Y_{a,b}:(\mathcal{A}\times{\mathbb{C}})^2\rightarrow (\mathcal{A}\times{\mathbb{C}})^2$, i.e. 
\begin{equation}\label{pYB-def}
    Y_{a,b}((x,a),(y,b))=((u((x,a),(y,b)),a),(v((x,a),(y,b)),b)),\quad a,b\in\mathbb{C},
\end{equation}
if it  solves the following equation:
\begin{equation}\label{YB-param}
Y^{12}_{a,b}\circ Y^{13}_{a,c} \circ Y^{23}_{b,c}=
Y^{23}_{b,c}\circ Y^{13}_{a,c} \circ Y^{12}_{a,b},
\end{equation}
where $a,b,c\in\mathbb{C}$.

\begin{remark}\label{triviallity}\normalfont
    A parametric Yang--Baxter map can be seen as a non-parametric Yang--Baxter map, if one  considers its parameters  as variables. That is, a parametric  Yang--Baxter map \eqref{pYB-def} can be understood as the (non-parametric) Yang--Baxter map:
    $$
    Y((x_1,x_2),(y_1,y_2))=(u((x_1,x_2),(y_1,y_2)),x_1,v((x_1,x_2),(y_1,y_2)),y_2)
    $$
\end{remark}

Parametric Yang--Baxter maps  will be denoted in short as:
\begin{equation}\label{YB-map-par}
    Y_{a,b}(x,y)=\big(u_{a,b}(x,y),\,v_{a,b}(x,y)\big),\quad x,y\in \mathcal{A},
\quad a,b\in\mathbb{C},
\end{equation}

A matrix ${\rm L}_a(x)={\rm L}_a(x,a,\lambda)$,  $x\in\mathcal{A}$, $a,\lambda\in\mathbb{C}$, is called a Lax matrix \cite{Suris} for a (parametric) Yang--Baxter map if it satisfies the following matrix refactorisation problem
\begin{equation}\label{eq-Lax}
{\rm L}_a(u){\rm L}_b(v)={\rm L}_b(y){\rm L}_a(x),\quad \text{for any}\quad\lambda\in\mathbb{C},
\end{equation}
which is actually the local 1-simplex equation. Equation \eqref{eq-Lax} is a generator of Yang--Baxter maps \cite{Suris, Kouloukas}.

It is believed in the literature that if a map \eqref{YB-map-par} is equivalent to a matrix refactorisation problem \eqref{eq-Lax}, then it is a Yang--Baxter map. This is  not true.

\begin{counterexample}\label{counterexample}
    Let  $\mathcal{A}=\mathbb{C}^2$. The map $Y:\mathbb{C}^2\rightarrow \mathbb{C}^2$ given by
    \begin{align*}
        x_1\mapsto u_1&=\frac{(a+x_1)(b-a+y_1)+(x_2+y_1) \left[(b-a) x_1+(a+x_1)y_1\right]}{a+x_1+(b+x_1) (x_2+y_1)}\\
        x_2\mapsto u_2&=y_2-\frac{(a-b)(b+y_1)}{a+x_1+(b+x_1) (x_2+y_1)}\\
         y_1\mapsto v_1&=x_1+\frac{a-b}{1+x_2+y_1}\\
           y_2\mapsto v_2&=\frac{(a+x_1) (a-b+x_2)+(x_2+y_1)\left[(a-b) (x_1-y_1)+x_2(b+x_1)\right]}{a+x_1+(b+x_1) (x_2+y_1)},
    \end{align*}
    is equivalent to the 1-local Yang--Baxter equation
   $$
{\rm L}_a(u_1,u_2){\rm L}_b(v_1,v_2)={\rm L}_b(y_1,y_2){\rm L}_a(x_1,x_2),
$$
where ${\rm L}_a(x_1,x_2)=\begin{pmatrix}
    a+x_1x_2+x_1 & x_1\\
    x_2 &1
\end{pmatrix}.$ However, map $Y$ is not a Yang--Baxter map.
\end{counterexample}

\subsection{Zamolodchikov tetrahedron and local 2-simplex equation}
A map $T:\mathcal{A}^3\rightarrow \mathcal{A}^3$, i.e. $T(x,y,z)=(u(x,y,z),v(x,y,z),w(x,y,z))$, $x,y,z\in\mathcal{A}$,
is called a Zamolodchikov  (tetrahedron) map  if it solves the equation
\begin{equation}\label{Tetrahedron-eq}
    T^{123}\circ T^{145} \circ T^{246}\circ T^{356}=T^{356}\circ T^{246}\circ T^{145}\circ T^{123}.
\end{equation}
This is the 3-simplex or Zamolodchikov tetrahedron equation.
Functions $T^{ijk}:\mathcal{A}^6\rightarrow  \mathcal{A}^6$, $i,j=1,2,3,~i\neq j$, act on the $ijk$ terms of the product $\mathcal{A}^6$ as map $T$, and trivially on the others. For instance,
$$
T^{356}(x,y,z,r,s,t)=(x,y,u(z,s,t),r,v(z,s,t),w(z,s,t)).
$$

Let  $a$, $b$ and $c$ be  complex parameters. A map $T:(\mathcal{A}\times\mathbb{C})^3\rightarrow (\mathcal{A}\times\mathbb{C})^3$, namely $T:((x,a),(y,b),(z,c))\mapsto ((u((x,a),(y,b),(z,c)),a),(v((x,a),(y,b),(z,c)),b),(w((x,a),(y,b),(z,c)),c))$, is called a parametric tetrahedron map if it solves the (parametric) equation:
\begin{equation}\label{Par-Tetrahedron-eq}
    T^{123}_{a,b,c}\circ T^{145}_{a,d,e} \circ T^{246}_{b,d,f}\circ T^{356}_{c,e,f}=T^{356}_{c,e,f}\circ T^{246}_{b,d,f}\circ T^{145}_{a,d,e}\circ T^{123}_{a,b,c}.
\end{equation}
Parametric 3-simplex maps will be denoted in short as:
\begin{equation}\label{Par-Tetrahedron_map}
 T_{a,b,c}:(x,y,z)\mapsto (u_{a,b,c}(x,y,z),v_{a,b,c}(x,y,z),w_{a,b,c}(x,y,z)),
\end{equation}

Now, let ${\rm L}={\rm L}(x,k,\lambda)$, $x\in\mathcal{A}$,  $k,\lambda\in\mathbb{C}$ be a matrix of the form
\begin{equation}\label{matrix-L}
   {\rm L}(x,k,\lambda)= \begin{pmatrix} 
a(x,k,\lambda) & b(x,k,\lambda)\\ 
c(x,k,\lambda) & d(x,k,\lambda)
\end{pmatrix},
\end{equation}
for scalar functions $a, b, c$ and $d$ (in general, they can  be matrices \cite{Korepanov}). Let ${\rm L}_{ij}$, $i,j=1,2, 3$, given by
{\small
\begin{equation}\label{Lij-mat}
   {\rm L}_{12}=\begin{pmatrix} 
 a(x,k,\lambda) &  b(x,k,\lambda) & 0\\ 
c(x,k,\lambda) &  d(x,k,\lambda) & 0\\
0 & 0 & 1
\end{pmatrix},\quad
 {\rm L}_{13}= \begin{pmatrix} 
 a(x,k,\lambda) & 0 & b(x,k,\lambda)\\ 
0 & 1 & 0\\
c(x,k,\lambda) & 0 & d(x,k,\lambda)
\end{pmatrix}, \quad
 {\rm L}_{23}=\begin{pmatrix} 
   1 & 0 & 0 \\
0 & a(x,k,\lambda) & b(x,k,\lambda)\\ 
0 & c(x,k,\lambda) & d(x,k,\lambda)
\end{pmatrix},
\end{equation}
}
where ${\rm L}_{ij}={\rm L}_{ij}(x,k,\lambda)$, $i,j=1,2,3$.

Then, ${\rm L}={\rm L}(x,k,\lambda)$ is a Lax matrix for map \eqref{Par-Tetrahedron_map}, if the latter satisfies \cite{Dimakis-Hoissen-2015}
\begin{equation}\label{Lax-Tetra}
    {\rm L}_{12}(u,a,\lambda){\rm L}_{13}(v,b,\lambda){\rm L}_{23}(w,c,\lambda)= {\rm L}_{23}(z,c,\lambda){\rm L}_{13}(y,b,\lambda){\rm L}_{12}(x,a,\lambda),\quad \text{for any}\quad\lambda\in\mathbb{C}.
\end{equation}
This is the \textit{local Yang--Baxter} equation or Maillet--Nijhoff equation \cite{Nijhoff} in Korepanov's form. In this matrix  form equation \eqref{Lax-Tetra} was suggested by Korepanov, and deserves to be referred as \textit{Korepanov's equation}. The Korepanov equation allows one to construct tetrahedron maps related to integrable systems of mathematical  physics \cite{Kassotakis-Tetrahedron, Sokor-2022} and `classify' tetrahedron  maps \cite{Kashaev-Sergeev}.

\section{Variation of the spectral parameter}\label{parameter variation}
A  Lax matrix for an $n$-simplex equation, ${\rm L}={\rm L}(x,k,\lambda)$, often depends on a spectral parameter $\lambda$, and the associated local $(n-1)$-simplex problems must  be  satisfied for any $\lambda\in\mathbb{C}$. The approach of variation of the parameter $k$ in ${\rm L}$ has been considered by other authors in the literature \cite{Doliwa-2014, Kassotakis-Nonlinearity, Kassotakis-Kouloukas}. Here, we shall variate the spectral parameter. The spectral parameter overdetermines the system, and, considering it as a variable  (or a function of  variables), we can reduce the number of  equations of the associated system of polynomial equations, whereas  we increase the number of its variables, which allows us to construct new $n$-simplex maps.

For example, it is known that substitution of  the Lax matrix ${\rm L}(x,a,\lambda)=\begin{pmatrix}x  & 1\\ x^2+a+\lambda, & x\end{pmatrix}$ to equation \eqref{eq-Lax}, one can derive the  famous Adler  map \cite{Adler} 
$$
Y_{a,b}:(x,y)\longmapsto \left(y-\frac{a-b}{x+y},x+\frac{a-b}{x+y}\right),
$$
which is  a parametric Yang--Baxter map.

We change the spectral parameter to $\lambda\rightarrow x_1 x_2$, that is:
$$
   {\rm L}_a(x_1,x_2)\equiv {\rm L}(x_1,x_2,a)=\begin{pmatrix}x_1  & 1\\ x_1^2+x_1x_2+a, & x_1\end{pmatrix},
$$
and substitute  it to \eqref{eq-Lax}:
$$
    {\rm L}_a(u_1,u_2){\rm L}_b(v_1,v_2)={\rm L}_b(y_1,y_2){\rm L}_a(x_1,x_2).
$$
The above is equivalent  to the system:
\begin{eqnarray*}
  &b + v_1 (u_1 + v_1 + v_2)=a + x_1 (y_1 + x_1 + x_2),&\\
  &u_1+v_1=x_1+y_1,&\\
  &u_1 \left[b + v_1 (v_1 + v_2)\right]+\left[a + u_1 (u_1 + u_2)\right] v_1 =
  y_1 \left[a + x_1 (x_1 + x_2)\right]+\left[b + y_1 (y_1 + y_2)\right] x_1,&\\
  &a + u_1(1 + u_2 + v_1)=b + y_1(1 + y_2 + x_1),&
\end{eqnarray*}
which has two solutions for $u_i, v_i$, $i=1,2$. The first is given by the parametric family of maps $Y_1:\mathbb{C}^4\rightarrow \mathbb{C}^4$
\begin{equation}\label{Adler_type}
    (x_1,x_2,y_1,y_2)\stackrel{Y_1}{\longrightarrow}\left(y_1,y_2-\frac{a-b}{y_1},x_1,x_2+\frac{a-b}{x_1}\right),
\end{equation}
whereas the second is presented by the parametric family of maps $Y_2:\mathbb{C}^4\rightarrow \mathbb{C}^4$
\begin{align}\label{Adler_type-2}
    x_1\mapsto u_1 &=y_1 - \frac{a - b + x_1 x_2 - y_1 y_2}{x_1 + y_1},\nonumber\\
    x_2\mapsto u_2 &=\frac{-x_1 x_2 (x_1 + y_1)}{a - b + y_1 (x_1 + y_1) - x_1 x_2 + y_1  y_2},\\
    y_1\mapsto v_1 &=x_1 + \frac{a - b + x_1 x_2 - y_1 y_2}{x_1 + y_1},\nonumber\\
    y_2\mapsto v_2 &=\frac{y_1 (x_1 + y_1) y_2}{a - b + x_1 (x_1 + y_1) + x_1 x_2 - y_1 y_2}.
\end{align}

Map $Y_1$ given in \eqref{Adler_type} admits the first integrals:
$$
I_1=x_1+y_1,\quad I_2=x_1y_1,\quad I_3=x_1x_2+y_1y_2,
$$
which are functionally independent. Moreover, the invariants $I_1$ and $I_2$ are in  involution with respect to the Poisson bracket:
$$
\left\{x_1,x_2\right\}=\left\{y_1,y_2\right\}=1,\quad \left\{x_i,y_j\right\}=0,\quad i,j=1,2.
$$
That is, map \eqref{Adler_type} is completely integrable.

On the other hand, map  $Y_2$ given in \eqref{Adler_type-2} shares the first integrals $I_1$ and $I_3$ with map $Y_1$. However, its Liouville integrability is an open problem.

We apply this idea to NLS type Darboux matrices.

\section{Darboux transformations and Yang--Baxter maps}\label{Yang--Baxter maps}
Darboux transformations are convenient candidates for variating the spectral parameter. In \cite{Sokor-Sasha}, we used the following Darboux matrices:
\begin{equation}\label{D-NLS-DNLS}
    {\rm M}(x_1,x_2,a)=\begin{pmatrix}
        a+x_1x_2+\lambda & x_1\\
        x_2 & 1
    \end{pmatrix}\quad\text{and}\quad
         {\rm N}(x_1,x_2,a)=\begin{pmatrix}
        \lambda^2(a+x_1x_2) & \lambda x_1\\
        \lambda x_2 & 1
    \end{pmatrix},
\end{equation}
where $x_i,y_i$, $i=1,2$, are complex variables and $a,\lambda$ are complex parameters, associated with the  NLS and  the derivative NLS equation, respectively, and we constructed Yang--Baxter maps.  Here, we replace the spectral parameter of matrices ${\rm M}$ and  ${\rm N}$  in \eqref{D-NLS-DNLS} by functions, and we consider their matrix refactorisation problems \eqref{eq-Lax}.

\subsection{NLS type Yang--Baxter maps}
Substitution of matrix ${\rm M_a}(x_1,x_2)\equiv{\rm M}(x_1,x_2,a)$ of \eqref{D-NLS-DNLS} to \eqref{eq-Lax} leads to the derivation of the Adler--Yamilov map \cite{Sokor-Sasha}. If we replace $\lambda+a\rightarrow x_3$ in ${\rm M}$, then \eqref{eq-Lax} will be  equivalent to a correspondence. If  we supplement this correspondence with  equations  
$$
u_2=y_2,\quad  u_3=x_3,\quad\text{and}\quad v_3=y_3,
$$
then we obtain the Yang--Baxter map
\begin{equation}\label{y_2}
    Y:(x_1,x_2,x_3,y_1,y_2,y_3)\longmapsto \left(y_1 -\frac{x_3-y_3}{1+x_1y_2}x_1,y_2,x_3,x_1,x_2+\frac{x_3-y_3}{1+x_1y_2}y_2,y_3\right),
\end{equation}
which is a trivial example, since it can be obtained from  the Adler--Yamilov map \cite{Sokor-Sasha}, as explained in Remark \ref{triviallity}.

Now, replace $\lambda\rightarrow x_3$ in ${\rm M}$, and consider  the matrix refactorisation problem
\begin{equation}\label{Lax-NLS}
    {\rm M}_a(u_1,u_2,u_3){\rm M}_b(v_1,v_2,v_3)={\rm M}_b(y_1,y_2,y_3){\rm M}_a(x_1,x_2,x_3),\quad  {\rm M}_a(x_1,x_2,x_3)=\begin{pmatrix}
        x_1x_2+x_3+a & x_1\\
        x_2 & 1
    \end{pmatrix}.
\end{equation}

The matrix equation \eqref{Lax-NLS}  is  equivalent to the  following system  of polynomial equations
\begin{eqnarray*}
   & u_1 v_2 + (a + u_1 u_2 + u_3) (b + v_1 v_2 + v_3)= y_1 x_2 + (b + y_1 y_2 + y_3) (a + x_1 x_2 + x_3),&\\
   &u_1 + (a + u_1 u_2 + u_3) v_1=y_1 + (b + y_1 y_2 + y_3) x_1&\\
   &v_2 + u_2 (b + v_1 v_2 + v_3)=x_2 + y_2 (b + x_1 x_2 + x_3),\\
   &u_2 v_1 = x_1 y_2,&
\end{eqnarray*}
which can be solved for $u_3, v_1, v_2$ and $v_3$:
\begin{subequations}\label{cor-NLS}
\begin{align}
    u_3&=\frac{x_1 (bu_2-ay_2)+u_2 \left[(y_1 - u_1) (1 + x_1 y_2) + x_1 y_3\right]}{x_1 y_2},\\
    v_1&=\frac{x_1y_2}{u_2},\\
    v_2&= x_2 + \frac{x_3 (u_1 - y_1) y_2}{(u_1 - y_1) (1 + x_1 y_2) - x_1 (b+y_3)},\\
    v_3&= \frac{x_1 x_3 y_2 y_3}{u_2 \left[(u_1 - y_1) (1 + x_1 y_2) - x_1 (b+y_3)\right]}
\end{align}
\end{subequations}

Equations \eqref{cor-NLS} define a correspondence between $\mathbb{C}^6$  and  $\mathbb{C}^6$. In order to define a map,  we need to supplement \eqref{cor-NLS} with two equations. 

\begin{remark}\normalfont
    The couterexample \ref{counterexample} was obtained by replacing $\lambda\rightarrow x_1$  in ${\rm M_a}(x_1,x_2)\equiv{\rm M}(x_1,x_2,a)$ of \eqref{D-NLS-DNLS}.
\end{remark}

Now, we shall obtain new,  nontrivial, Yang--Baxter maps from \eqref{cor-NLS}. We have the following. 

\begin{theorem}
 Map $Y_3:\mathbb{C}^6\rightarrow \mathbb{C}^6$ given by
\begin{equation}\label{y_2}
    Y_3:(x_1,x_2,x_3,y_1,y_2,y_3)\longmapsto \left(y_1, \frac{ a + y_3}{b + y_3}y_2,y_3,\frac{b+y_3}{a+y_3}x_1,x_2,\frac{x_3(b+y_3)+(a-b)y_3}{a+y_3} \right)
\end{equation}
 and map $Y_4:\mathbb{C}^6\rightarrow \mathbb{C}^6$ given by
\begin{subequations}\label{y_3}
    \begin{align}
         x_1&\mapsto u_1=\frac{x_3 (a+x_3)y_1y_2+x_2 (x_3-y_3)\left[x_1 (b+y_3)+y_1 (1+x_1y_2)\right]}{x_3(a+x_3)y_2+x_2(1+x_1y_2)(x_2-y_3)},\\
         x_2&\mapsto u_2=\frac{x_3(a+x_3)y_2+x_2(1+x_1y_2)(x_2-y_3)}{x_3 (b+y_3)},\\
         x_3&\mapsto u_3=x_3,\\
         y_1&\mapsto v_1=\frac{x_1x_3y_2(b+y_3)}{x_3(a+x_3)y_2+x_2(1+x_1y_2)(x_2-y_3)},\\
         y_2&\mapsto v_2=\frac{x_2 y_3}{x_3},\\
         y_3&\mapsto v_3=y_3,
   \end{align}
\end{subequations}
are noninvolutive parametric Yang--Baxter maps which have a common Lax representation \eqref{Lax-NLS}. Moreover, map $Y_3$ has the functionally independent first integrals
\begin{equation}\label{y_2ints}
    I_1=x_1y_2+x_2y_1,\quad I_2=(a+x_1x_2+x_3)(b+y_1y_2+y_3)
\end{equation}
while  map $Y_4$ admits the functionally independent first integrals
\begin{equation}\label{y_3ints}
    I_1=x_1y_2+x_2y_1+(a+x_1x_2+x_3)(b+y_1y_2+y_3),\quad I_2=x_3y_3,\quad I_3=x_3+y_3.
\end{equation}
\end{theorem}
\begin{proof}
    For the correspondence \eqref{cor-NLS} to define a map we need two more (independent) equations which will define the two remaining free variables. Supplementing \eqref{cor-NLS} with equations
    $$
    u_3=y_3,\quad\text{and}\quad v_2u_3=x_2y_3,
    $$
    and solving the augmented system for  $u_i$ and $v_i$, $i=1,2,3$, we obtain map $Y_3$ given by \eqref{y_2}. 

    Then, we have that
    $$
    I_ i\circ Y_2 = I_i,\quad i=1,2,
    $$
    for $I_i$, $i=1,2$, given by \eqref{y_2ints}. That is, $I_i$, $i=1,2$, from \eqref{y_2ints} are invariants of the map $Y_2$. Moreover, it can be verified that $\nabla I_i$, $i=1,2$, are linearly independent, namely  $I_i$, $i=1,2$, are functionally independent.
    
    Moreover, supplementing system \eqref{cor-NLS} with equations
   $$
    u_3=x_3,\quad\text{and}\quad v_2u_3=x_2y_3,
    $$
    we obtain map $Y_4$ given  by \eqref{y_3}. The Yang--Baxter  equation can  be verified by substitution to \eqref{eq_YB}.

    Now, functions $I_2$ and $I_3$ are obvious first integrals of  $Y_4$, and $I_1\circ Y_3=I_1$. Moreover, matrix $(\nabla I_i)_i$ has rank 3, that is $I_i$, $i=1,\ldots,3$, from \eqref{y_3ints}, are  functionally independent.

    Finally,  
    $$
    u_3\circ Y_3=\frac{x_3(b+y_3)+(a-b)y_3}{a+y_3}\neq  x_3,\quad\text{and}\quad v_2\circ Y_4=\frac{\left[x_3 (a + x_3) y_2 + x_2 (1 + x_1 y_2) (x_3 - y_3)\right]y_3}{(b+y_3)x_3^2}\neq y_3,
    $$
    which means that $Y_i$, $i=3,4$, are noninvolutive.
\end{proof}

\begin{remark}\normalfont
    The first integrals \eqref{y_2ints} of map $Y_3$ are in involution with  respect to the Poisson bracket
    $$
    \{x_1,y_2\}=\{y_1,y_2\}=1,\quad\text{and all the rest}\quad\{x_i,x_j\}=\{y_i,y_j\}=\{x_i,y_i\}=0,
    $$
    which is invariant under $Y_3$. However, one more first integral must be found in order to claim Liouville integrability. The invariants \eqref{y_3ints} of map $Y_4$ indicate its integrability. Its Liouville integrability is an open problem.
\end{remark}

\subsection{Derivative NLS type Yang--Baxter maps}
Here, we change $\lambda\rightarrow x_3$ in ${\rm N}$ of \eqref{D-NLS-DNLS}, and we substitute it into the Lax  equation 
\begin{align}\label{Lax-DNLS}
    &{\rm N}_a(u_1,u_2,u_3){\rm N}_b(v_1,v_2,v_3)={\rm N}_b(y_1,y_2,y_3){\rm N}_a(x_1,x_2,x_3),\\
    &{\rm N}_a(x_1,x_2,x_3)\equiv {\rm N}(x_1,x_2,x_3,a)=\begin{pmatrix}
        x_3^2(a+x_1x_2) &  x_1x_3\\
         x_2 x_3 & 1
    \end{pmatrix}.
\end{align}
Equation \eqref{Lax-DNLS} is equivalent to the system of polynomial equations
\begin{eqnarray}\label{sysDNLS}
    &u_1 u_3 v_2 v_3 + (a + u_1 u_2) (b + v_1 v_2) u_3^2v_3^2 = 
  x_2 x_3 y_1 y_3 + (a + x_1 x_2)  (b + y_1 y_2) x_3^2 y_3^2,&\\ 
 &u_1 u_3 + (a + u_1 u_2) u_3^2 v_1 v_3 = y_1 y_3 + 
   (b + y_1 y_2) y_3^2x_1 x_3,& \\
   &v_2 v_3 + u_2 u_3 (b + v_1 v_2) v_3^2 = 
  x_2 x_3 + (a + x_1 x_2) x_3^2 y_2 y_3,&\\ 
  &u_2 u_3 v_1 v_3 = x_1 x_3 y_2 y_3.&
\end{eqnarray}

As  in the previous section, this system is under-determined and thus it defines a correspondence rather than a  map. Indeed, it can be solved  for $u_1,u_2,u_3,v_2$ and $v_1$, $v_3$ will be  free variables. In order to obtain a map, we need to supplement the system with two more equations.

We have the following.

\begin{theorem}
 Map $Y_5:\mathbb{C}^6\rightarrow \mathbb{C}^6$ given by
 \begin{equation}\label{y_4}
         Y_5:(x_1,x_2,x_3,y_1,y_2,y_3)\longmapsto \left(y_1 -\frac{(a-b)x_1 x_3 y_3}{1 + x_1 x_3 y_2 y_3} ,y_2,y_3,x_1,x_2+\frac{(a-b)x_3 y_2 y_3}{1 + x_1 x_3 y_2 y_3},x_3\right),
 \end{equation}   
 and map $Y_6:\mathbb{C}^6\rightarrow \mathbb{C}^6$ given by
\begin{subequations}\label{Y5}
    \begin{align*}
        x_1&\mapsto u_1=y_1-\frac{(a - b)x_1 x_3 y_3}{1 - x_3 (a - x_1 y_2) y_3},\\
        x_2&\mapsto u_2=\frac{1 - a x_3 y_3 + x_1 x_3 y_2 y_3}{1 - b x_3 y_3 + x_1 x_3 y_2 y_3}y_2,\\
        x_3&\mapsto u_3=y_3,\\
        y_1&\mapsto v_1=\frac{1 - b x_3 y_3 + x_1 x_3 y_2 y_3}{1 - a x_3 y_3 + x_1 x_3 y_2 y_3}x_1,\\
        y_2&\mapsto v_2=x_2+\frac{(a - b)x_3y_2y_3}{1 - x_3 (b - x_1 y_2) y_3},\\
        y_3&\mapsto v_3=x_3
    \end{align*}
\end{subequations}
are noninvolutive parametric Yang--Baxter maps which have a common Lax representation \eqref{Lax-NLS}. Moreover, map $Y_5$ admits the functionally independent first integrals
\begin{equation}\label{y_4ints}
I_1=  (x_1 y_2+x_2 y_1) x_3y_3 + (a + x_1 x_2)  (b + y_1 y_2) x_3^2 y_3^2,\quad I_2=x_3+y_3,\quad I_3=x_3y_3,
\end{equation}
whereas map $Y_6$ admits the functionally independent first integrals
\begin{equation}\label{Y5ints}
I_1=  (x_1 y_2+x_2 y_1) x_3y_3 + (a + x_1 x_2)  (b + y_1 y_2) x_3^2 y_3^2,\quad I_2=x_3+y_3,\quad I_3=x_3y_3,\quad I_4=x_1+y_1,\quad I_5=x_2+y_2.
\end{equation}
\end{theorem}
\begin{proof}
If  we supplement system \eqref{sysDNLS} with equations
$$
v_1=x_1\quad\text{and}\quad v_3=x_3,
$$
then the associated augmented system admits two solutions, namely the following
$$
u_1=-y_1 -\frac{(a-b)x_1 x_3 y_3}{1 + x_1 x_3 y_2 y_3},\quad u_2=y_2,\quad u_3=-y_3,\quad v_1=x_1,\quad v_2=x_2+\frac{(a-b)x_1 x_3 y_3}{1 + x_1 x_3 y_2 y_3},\quad v_3=x_3,
$$
which  is not a Yang--Baxter map, and also the solution given  by map \eqref{y_4} which is a parametric Yang--Baxter map that  can  be readily verified by substitution to equation \eqref{YB-param}.

Furthermore, if  we supplement system \eqref{sysDNLS} with equations
$$
u_1+v_1=x_1+y_1,\quad v_3=x_3,
$$
then the corresponding  augmented system has a  unique  solution given by $Y_6$ in \eqref{Y5}. It can  be proved that  $Y_5$ is a parametric Yang--Baxter map by substitution to \eqref{YB-param}.

Maps $Y_5$ and $Y_6$ share the first integral $I_1$ which follows from the trace of the monodromy matrix, $\tr({\rm N}_b(y_1,y_2,y_3){\rm N}_a(x_1,x_2,x_3))$. The rest of invariants are obvious.

Noninvolutivity follows from the fact that $Y_i\circ Y_5\neq\id$, $i=5,6$.
\end{proof}

The invariants of $Y_5$ and $Y_6$ indicate their integrability. 

\begin{theorem}
    Map $Y_6$ is completely (Liouville) integrable.
\end{theorem}
\begin{proof}
     The Poisson bracket:
    $$
    \{x_1,x_2\}=\{y_1,y_2\}=-1,\quad\{x_1,y_2\}=\{y_1,x_2\}=1,
    $$
    where all the rest entries of the associated Poisson matrix are 0, is invariant under the map $Y_6$. The corresponding $6\times 6$ Poisson matrix ${\rm P}$ has rank 2, and $I_i$, $i=2,3,4,5$, are 4 Casimir  functions, since  $(\nabla I_i)\cdot {\rm  P}=0$, $i=2,3,4,5$. Thus, map $Y_2$ is completely integrable in the Liouville  sense.
\end{proof}

\section{Darboux transformations and Zamolodchikov maps}\label{tetrahedron maps}
In \cite{Sokor-2020-2} we constructed Zamolodchikov 3-simplex maps by substituting matrices \eqref{D-NLS-DNLS} into \eqref{Lax-Tetra}. In that construction, we had fixed the  spectral parameter $\lambda =0$, since  matrices \eqref{D-NLS-DNLS} do not solve equation \eqref{Lax-Tetra} for $\lambda\neq 0$. This, possibly, has  to do  with  the  fact that Zamolodchikov tetrahedron maps are related  to 3D lattice equations which  are discrete versions of $(2+1)$-integrable PDEs, where the  spectral parameter is usually absent. 

Here, instead of fixing $\lambda =0$, we shall variate it. Specifically, after changing in \eqref{D-NLS-DNLS} $\lambda\rightarrow x_3$, namely we set ${\rm M}^a(x_1,x_2,x_3)=\begin{pmatrix} x_1x_2+x_3+a & x_1\\ x_2 & 1\end{pmatrix}$, we consider the  $3\times  3$  generalisations of ${\rm M}$:
\begin{align*}
      {\rm M}^a_{12}(x_1,x_2,x_3)&=\begin{pmatrix} 
 x_1x_2+x_3+a &  x_1 & 0\\ 
x_2 &  1 & 0\\
0 & 0 & 1
\end{pmatrix},\\
 {\rm M}^a_{13}(x_1,x_2,x_3)&= \begin{pmatrix} 
 x_1x_2+x_3+a & 0 & x_1\\ 
0 & 1 & 0\\
x_2 & 0 & 1
\end{pmatrix}, \\
 {\rm M}^a_{23}(x_1,x_2,x_3)&=\begin{pmatrix} 
   1 & 0 & 0 \\
0 & x_1x_2+x_3+a & x_1\\ 
0 & x_2 & 1
\end{pmatrix},
\end{align*}
and substitute them  to equation
$$
{\rm M}^a_{12}(u_1,u_2,u_3){\rm M}^b_{13}(v_1,v_2,v_3){\rm M}^c_{23}(w_1,w_2,w_3)= {\rm M}^c_{23}(z_1,z_2,z_3){\rm M}^b_{13}(y_1,y_2,y_3){\rm M}^a_{12}(x_1,x_2,x_3).
$$

The latter is equivalent to the system
\begin{subequations}\label{tetra-corr}
    \begin{eqnarray}
    &(a + u_1 u_2 + u_3) (b + v_1 v_2 + v_3) = (a + x_1 x_2 + x_3) (b + y_1 y_2 + 
     y_3),&\\
     &(a + u_1 u_2 + u_3) v_1 w_2 + u_1 (c + w_1 w_2 + w_3) = 
  x_1 (b + y_1 y_2 + y_3),& \\
  &u_2 (b + v_1 v_2 + v_3) = (a + x_1 x_2 + x_3) y_2 z_1 + x_2 (c + z_1 z_2 + z_3),&\\ 
  &(a + u_1 u_2 + u_3) v_1 + u_1 w_1 =  y_1,&\\
  &(u_2 v_1 + w_1) w_2 + w_3 = (x_1 y_2 + z_2) z_1 + z_3,&\\
 &u_2 v_1 + w_1 = z_1,&\\
 &v_2 = (a + x_1 x_2 + x_3) y_2 + x_2 z_2,&\\
 &w_2 = x_1 y_2 + z_2,&
\end{eqnarray}
\end{subequations}
which can be solved for $u_1,u_2,v_1,v_2,v_3,w_1,w_2,w_3$ in terms  of $u_3$, namely it  defines a correspondence between $\mathbb{C}^9$ and $\mathbb{C}^9$. We omit writing the correspondence explicitly because of its length. However, in order to obtain  a  map, we need one more equation for  $u_3$.

\begin{theorem}
    Map $T_{a,b,c}:\mathbb{C}^9\rightarrow\mathbb{C}^9$ given  by
    \begin{subequations}\label{tetraNLS}
    \begin{align}
        x_1&\mapsto u_1=\frac{x_1 (b + y_3) - y_1 z_2}{c + z_3},\\
        x_2&\mapsto u_2=\frac{(a + y_3) (c + z_3) \left[(a + x_1 x_2 + x_3) y_2 z_1 + (c + z_1 z_2 + z_3) x_2\right]}{\mathcal{A}},\\
        x_3&\mapsto u_3=y_3,\\
        y_1&\mapsto v_1=\frac{y_1 (c + z_1 z_2 + z_3)-(b + y_3) x_1z_1 }{(a + y_3) (c + z_3)},\\
        y_2&\mapsto v_2 = (a + x_1 x_2 + x_3) y_2 + x_2 z_2,\\
        y_3&\mapsto v_3=\frac{b (x_3 - y_3) + (a + x_3) y_3}{a + y_3},\\
        z_1&\mapsto w_1=\frac{(c+z_3) \left[(a+x_1x_2+x_3)(b+y_3)z_1-(c+z_1z_2+z_3)x_2y_1\right]}{\mathcal{A}},\\
        z_2&\mapsto w_2=x_1 y_2 + z_2,\\
        z_3&\mapsto w_3=z_3,\label{tetraNLS-i}
    \end{align}
    \end{subequations}
    where
    \begin{align*}
        \mathcal{A}=&a b c + b c x_3 + a c y_1 y_2 + c x_1 x_2 y_1 y_2 + c x_3 y_1 y_2 + a c y_3 +c x_3 y_3 - a b x_1 y_2 z_1 - b x_1^2 x_2 y_2 z_1 -\\
         & b x_1 x_3 y_2 z_1 - a x_1 y_2 y_3 z_1 -x_1^2 x_2 y_2 y_3 z_1 - x_1 x_3 y_2 y_3 z_1 + c x_2 y_1 z_2 -  b x_1 x_2 z_1 z_2 + a y_1 y_2 z_1 z_2 + \\
         &x_1 x_2 y_1 y_2 z_1 z_2 + x_3 y_1 y_2 z_1 z_2 - 
         x_1 x_2 y_3 z_1 z_2 + x_2 y_1 z_1 z_2^2 + a b z_3 + b x_3 z_3 + a y_1 y_2 z_3 + \\
 & x_1 x_2 y_1 y_2 z_3 + x_3 y_1 y_2 z_3 +a y_3 z_3 + x_3 y_3 z_3 + x_2 y_1 z_2 z_3,
    \end{align*}
    is  a parametric Zamolodchikov tetrahedron map, and admits  the following functionally independent first integrals
    $$
    I_1=(a + x_1 x_2 + x_3) (b + y_1 y_2 + y_3),\quad I_2=(a+x_3) (b+y_3),\quad I_3=z_3.
    $$
\end{theorem}
\begin{proof}
    We supplement system \eqref{tetra-corr} with equation $u_3=y_3$, and the corresponding augmented system has a unique  solution given by \eqref{tetraNLS}. This is a parametric  parametric Zamolodchikov tetrahedron map, which  by be checked with  straightforward substitution to \eqref{Par-Tetrahedron-eq}.

    Regarding  the first integrals, we have that:
    \begin{eqnarray*}
    &(a + u_1 u_2 + u_3) (b +v_1 v_2 + v_3)\stackrel{\eqref{tetraNLS}}{=} (a + x_1 x_2 + x_3) (b + y_1 y_2 + y_3)&\\
    &(a+u_3) (b+v_3)\stackrel{\eqref{tetraNLS}}{=} (a+x_3) (b+y_3)&\\
    &w_3 \stackrel{\eqref{tetraNLS-i}}{=} z_3,&
    \end{eqnarray*}
    and matrix $(\nabla I_i)_i$ has rank 3, therefore $I_i$, $i=1,2,3$, are functionally independent.
\end{proof}

\begin{remark}\normalfont
In matrix ${\rm M}^a(x_1,x_2,x_3)=\begin{pmatrix} x_1x_2+x_3+a & x_1\\ x_2 & 1\end{pmatrix}$, which we used to construct the parametric Zamolodchikov tetrahedron map \eqref{tetraNLS}, we kept the parameter $a$. At first glance, it seems that the dependence on parameter $a$ is trivial, since we could change $x_3\rightarrow a+x_3$. That would  mean that in the derived tetrahedron map one could eliminate the parameters $a,b$ and  $c$,  by changing in $T_{a,b,c}$
$$
(x_3,y_3,z_3)\rightarrow (x_3-a, y_3-b, z_3-c)\quad\text{and}\quad (u_3,v_3,w_3)\rightarrow (u_3-a, v_3-b, w_3-c).
$$
However, this  is  not the case for map \eqref{tetraNLS}, because we supplemented the correspondence, which is equivalent to system \eqref{tetra-corr}, with the equation $u_3=y_3$ that breaks this symmetry.
\end{remark}

\section{Discussion}\label{conclusions}
We presented a method for constructing hierarchies of 2- and 3-simplex maps by variating the spectral parameter in their Lax representations. We used this method to construct new such maps.

In particular, we constructed an Adler type Yang--Baxter maps $Y_1$ (map \eqref{Adler_type})  and $Y_2$ (map \eqref{Adler_type-2}), two NLS type Yang--Baxter maps $Y_3$ (map \eqref{y_2}) and $Y_4$ (map \eqref{y_3}), two derivative  NLS  type parametric Yang--Baxter maps $Y_5$ (map \eqref{y_4}) and  $Y_6$ (map \eqref{Y5}), and a parametric Zamolodchikov tetrahedron map $T_{a,b,c}$ (map \eqref{tetraNLS}).  We proved that  maps $Y_1$ and $Y_6$ are completely (Liouville) integrable.

Our results can be extended  in the following  ways:

\begin{enumerate}
    \item Complete integrability of $Y_2$, $Y_3$, $Y_4$, $Y_5$ and $T_{a,b,c}$.

    For map $Y_3$ one needs one more invariant. Maps $Y_2$, $Y_4$ and $Y_5$ have enough first integrals for  integrability aims. For their Liouville integrability one needs to find a  certain Poisson bracket. For map $T_{a,b,c}$ two more functionally independent first integrals are  needed.
    
    \item Noncommutative versions of all  the  derived maps.
    
    Noncommutative versions and extensions of $n$-simplex maps is a hot topic for  research (see,  e.g. \cite{Giota-Miky, Doliwa-2014,  Doliwa-Kashaev, GKM, Kassotakis-Nonlinearity, Kassotakis-Kouloukas, Sokor-2022},  and the reference therein). One could construct noncommutative versions of all the derived maps assuming that the variables in their Lax matrices are elements of a noncommutative division ring. In  order to prove that the maps are indeed $n$-simplex maps we need to  use certain matrix  refactorisation properties \cite{Sokor-2022, Sokor-Nikitina,  Kouloukas}. The same idea can  be applied to all the Grassmann extended Yang--Baxter maps presented  in \cite{Sokor-2020, Sokor-Kouloukas, Sokor-Sasha-2}.

    \item Extend to $4$-simplex maps.
    
    Using  the methods of \cite{SKR-2023-PhysD, SKR-JHEP}, one  could  extend all  the derived maps to higher members of  the family of  $n$-simplex equations.

    \item Construction of $n$-gon maps.

    The $n$-simplex maps are related to the solutions  of  the $n$-gon  equation \cite{EKT-2024,  Korepanov-Kolmogorov, Kassotakis-Proc}. One  could use  the matrices of  this paper as potential  generators of solutions to the pentagon equation \cite{Drinfeld-pentagon}.
\end{enumerate}

\section*{Acknowledgements}
The work on  sections 2, 3 and 4 was funded  by the Russian Science Foundation, project No. 20-71-10110 (https://rscf.ru/en/project/23-71-50012/). The work on sections 1, 5 and 6 was supported by the Russian Ministry of Science and Higher Education (Agreement No. 075-02-2024-1442). I would like to thank Pavlos Kassotakis for several useful conversations.

\section*{Conflict of Interest}
The authors declare that they have no conflicts of interest.

\end{document}